%% file: main.tex
\documentclass{ifacconf}
\usepackage{url}
\usepackage{etoolbox}
\usepackage{graphicx}
\usepackage{natbib}           
\usepackage{amsmath,amssymb,bm}
\usepackage{mathtools}        
\usepackage{booktabs}         
\usepackage{siunitx}          

\makeatletter
\@ifpackageloaded{hyperref}{%
  \pdfstringdefDisableCommands{%
    \def\bm#1{#1}%
    \def\mathbf#1{#1}%
    \def\mathbb#1{#1}%
    \def\mathrm#1{#1}%
  }%
}{}
\makeatother
\makeatletter
\@ifundefined{theoremstyle}{}{\let\theoremstyle\relax}
\makeatother
\usepackage{amsthm}

\newtheorem{theorem}{Theorem}

\newtheorem{definition}{Definition}
\newtheorem{assumption}{Assumption}

\newtheorem{proposition}{Proposition}
\newtheorem{remark}{Remark}
\newtheorem{hypothesistesting}{Hypothesis Test}
\newtheorem{problem}{Problem}
\usepackage{algorithm}
\usepackage[noend]{algpseudocode}   

\usepackage{tikz}
\usetikzlibrary{arrows.meta,shapes,positioning,fit,calc}

\newcommand{\N}{\mathcal{N}}
\newcommand{\Prob}{\mathbb{P}}
\newcommand{\E}{\mathbb{E}}

\usepackage{pgfplots}
\pgfplotsset{compat=1.18}
\usepackage{pgfplotstable}
\usepgfplotslibrary{fillbetween}
\usetikzlibrary{patterns}
\begin{document}

\begin{frontmatter}

\title{An Innovation-Based Approach to Detect Stealthy Disturbance Attacks in Maritime Monitoring\thanksref{note1}}

\author[ucbm,NITEL,corr]{Gabriele Oliva}
\author[ucbm,NITEL]{Bianca Mazz\'a}
\author[ucbm,NITEL]{Roberto Setola}
\address[ucbm]{Universit\`a Campus Bio-Medico di Roma. Via A. del Portillo 21, 00128 Rome, Italy. }
\address[NITEL]{Consorzio Nazionale Interuniversitario per i Trasporti e la Logistica (NITEL), via Urbino, 31, 00182 Rome, Italy.}
\address[corr]{Corresponding author. Email: \texttt{g.oliva@unicampus.it}}
\thanks[note1]{This work was partly supported by project VIGIMARE, funded by the European Union under grant n. 101168016. Views and opinions expressed are however those of the authors only and do not necessarily reflect those of the European Union. Neither the European Union nor the granting authority can be held responsible for them. This work was partly supported by Italian National project IMPROVE, funded by the Italian Ministry of Defense under grant n. 20711.  
}
\begin{abstract}
Modern maritime navigation and control systems rely on digital sensing, estimation, and communication pipelines that fuse GNSS, radar, inertial, and AIS data through approaches such as Kalman-filter-based estimators.
While these technologies are essential for safety and efficiency, their growing interconnection also exposes vessels to faults and cyber–physical anomalies.
This paper introduces a \emph{Statistical Detection Suite} (SDS) to detect malicious stealthy disturbances. Specifically, the SDS operates directly on the innovations of Kalman filters, providing a lightweight yet statistically grounded layer of anomaly monitoring within maritime estimation frameworks.
The SDS jointly evaluates whitened innovations through four complementary checks: (i) bias, (ii) covariance consistency via the normalized innovation squared (NIS), (iii) Gaussianity, and (iv) temporal independence via portmanteau statistics.
The analysis further examines how an adversary can craft stealthy finite-impulse-response (FIR) Gaussian disturbances that can evade classical $\chi^2$ checks, formulating an optimization-based design that balances stealth and trajectory impact.
An evaluation in maritime navigation scenarios illustrates how the SDS exposes colored spoofing attacks that bypass traditional methods, highlighting the role of innovation-based monitoring in strengthening maritime resilience against cyber–physical threats.
\end{abstract}


\end{frontmatter}

\section{Introduction}

\color{black}
Recent reports highlight the \textcolor{black}{growing} exposure of maritime systems to cyber threats:
detecting a cyber attack on a vessel may take several months~\citep{MITAGSGuide}, while
reported maritime cyber incidents \textcolor{black}{surged} from 10 in 2021 to at least 46 in
2025~\citep{NHLStenden2025}.
These trends motivate the need for adaptive monitoring methods capable of revealing
subtle faults and attacks within navigation pipelines.

Modern maritime navigation relies on sensor-fusion architectures that integrate GNSS,
inertial sensors, radar, and the Automatic Identification System (AIS) within unified
estimators, typically based on the Kalman filter (KF).
Because such filters combine heterogeneous and potentially untrusted data sources, their
statistical consistency is directly tied to safety and cybersecurity, making detection
\emph{within} the estimation loop essential.

Most maritime anomaly-detection methods operate at the trajectory or data level, ranging
from analytical techniques to deep-learning classifiers~\citep{riveiro2018,forti2021,capobianco2023}. \textcolor{black}{Comparatively} few target the estimation layer where deviations first arise.
Innovation-based monitoring addresses this gap: under nominal conditions, KF innovations
are white Gaussian with zero mean and unit covariance. \textcolor{black}{Departures} from these
properties provide statistical evidence of faults, spoofing, or model mismatch.

\textcolor{black}{Building on this principle, a \emph{Statistical Detection Suite} (SDS) is proposed that jointly tests the mean, covariance, Gaussianity, and temporal independence of the innovations.}

Innovation analysis has long been used as a diagnostic tool in system identification and
navigation~\citep{bergman1985,Bierman}. \textcolor{black}{Extensions have improved} sensitivity in
structural health monitoring~\citep{bernal2011} and \textcolor{black}{in} bounding finite-time risks in
GNSS-integrated systems~\citep{marti2025kalman}.
Surveys of maritime anomaly detection further emphasize the need for interpretable \textcolor{black}{and}
statistically grounded methods compatible with navigation pipelines~\citep{riveiro2018}.

\noindent {\bf Motivation and gap.}
Most maritime anomaly-detection approaches operate at the trajectory level (e.g., pattern mining, clustering, or deep predictors applied to AIS/GNSS tracks)~\citep{riveiro2018,forti2021,capobianco2023}.
While effective for surveillance, these methods typically detect anomalies \emph{after} they have manifested at the track level and often require nontrivial computational resources and training data.
In contrast, modern bridge and autonomy stacks increasingly rely on KF-based estimators that fuse heterogeneous and potentially untrusted sensors (GNSS, inertial, radar, AIS).
In such pipelines, anomalies first appear as statistical inconsistencies in the KF innovation sequence.
This motivates \textcolor{black}{the development of a lightweight and interpretable} \emph{estimation-layer} monitoring mechanism, directly compatible with existing KF implementations.

\noindent {\bf State of the art.}
Innovation testing has a long history in KF validation and diagnosis:
whiteness and $\chi^2$ consistency checks are standard tools in navigation and tracking~\citep{Bierman}, with variants and refinements proposed across application domains~\citep{bergman1985,bernal2011}.
Recent integrity-monitoring work exploits innovation-based constructions to bound risks in GNSS-integrated systems~\citep{marti2025kalman}.
In parallel, the cyber--physical security literature has shown that an adversary can craft stealthy perturbations that respect classical $\chi^2$/NIS constraints while inducing damage, including optimal deception under detection constraints~\citep{guo2017,guo2018} and worst-case impact analyzes via information-theoretic metrics~\citep{milosevic2018,bai2017}.
Detector-side enhancements often remain centered on \emph{single-property} monitoring, most commonly \textcolor{black}{on} covariance/NIS (possibly with windowing or CUSUM-like accumulation)~\citep{tunga2017}.
However, attacks and realistic mismatches in maritime navigation may violate \emph{other} defining properties of the nominal innovations, such as temporal independence (coloring) or distributional shape, which a pure NIS detector can miss.

Alongside these KF-centric works, maritime anomaly detection has also progressed through Bayesian filters~\citep{forti2021} and learning-based predictors~\citep{capobianco2023,zuzic2025}, and GNSS spoofing detection can be addressed via indirect filtering schemes~\citep{jin2025indirect}.
These approaches are complementary, but they do not provide a unified, per-property innovation diagnostic suite designed to operate \emph{inside} KF-based maritime estimation loops and to explicitly expose temporally colored stealthy disturbances.

\color{black}
A related line of work considers innovation-based attacks on Kalman filtering systems, where the adversary manipulates the innovation sequence to degrade estimation performance while remaining stealthy, typically under a Kullback-Leibler divergence constraint (e.g., \cite{guo2018worst,zhou2022optimal}). 
However, these approaches focus on measurement-level attacks in remote estimation architectures and define optimality with respect to specific stealthiness metrics and detector structures. In contrast, the present work considers process-level perturbations and targets a broader set of properties beyond innovation statistics. As a result, the attack models and optimality notions are not directly comparable.
\color{black}

\noindent {\bf Contribution and novelty.}
This paper contributes an innovation-based SDS tailored to KF-based maritime monitoring, and an adversarial design framework that clarifies the damage--stealth trade-off under multi-test monitoring:
\begin{itemize}\setlength{\itemsep}{1pt}
\item An SDS is proposed that \emph{jointly} tests four defining properties of whitened innovations under nominal KF operation: (i) zero mean (bias), (ii) covariance consistency via NIS, (iii) Gaussianity (shape), and (iv) temporal independence via a multivariate portmanteau statistic. Unlike detectors centered solely on $\chi^2$/NIS, the SDS reports per-test statistical evidence, improving interpretability and sensitivity to colored disturbances.
\item An attacker model based on FIR filtering of a Gaussian seed \textcolor{black}{is introduced}, which preserves mean and Gaussianity by construction while allowing the attacker to shape covariance and temporal correlations to evade covariance/NIS and whiteness tests.
\item The attacker design is formulated as a constrained optimization and \textcolor{black}{derives} a tractable convex relaxation over the disturbance covariance, yielding an upper bound on achievable damage. Then, a constructive relaxation--recovery--certification pipeline is provided that produces \emph{implementable} FIR disturbances (via structured factorization/projection) together with explicit scaling laws that ensure satisfaction of SDS constraints after recovery.
\item The SDS and the attack trade-off are demonstrated in a representative maritime navigation scenario, showing that temporally colored attacks can bypass NIS-only monitoring but are exposed by whiteness testing, while optimized FIR-Gaussian disturbances can retain stealth with limited yet nonzero trajectory impact.
\end{itemize}

\color{black}

\noindent {\bf Preliminaries.}
Bold lowercase/uppercase denote vectors/matrices.
For a matrix $X$, $\|X\|_2$, $\|X\|_F$ and $\|X\|_*$ are the spectral, Frobenius and nuclear norm, respectively. Moreover, for a vector ${\bm x}$, $\|{\bm x}\|_2$ is the Euclidean norm.
$\mathrm{tr}(X)$ is the trace, $X^\top$ the transpose, and $X\succeq 0$ means $X$ is positive semidefinite (PSD).
$\mathcal N(\mu,\Sigma)$ is a Gaussian distribution with mean $\mu$ and covariance $\Sigma$.
Given a random vector $z$, $\E[z]$ and $\mathrm{Cov}(z)$ denote expectation and covariance.
To denote the Frobenius inner product of two matrices the notation $\langle A,B\rangle$ is used, i.e., $\langle A,B\rangle=\mathrm{tr}(A^\top B)$.
\section{Problem Setting}\label{sec:setup}
Let us consider a linear time-invariant (LTI) system with zero mean i.i.d. Gaussian process and measurement noise. Moreover, let us assume the presence of a latent exogenous input (fault/attack) that is zero under nominal conditions\footnote{\textcolor{black}{The fault/attack is modeled at the measurement level and is independent of the specific sensing modality (e.g., GNSS, AIS, radar).}}. Overall, the system's dynamics is given by
\begin{equation}
\begin{cases}
   {\bm x}[k+1] &= A{\bm x}[k] + B {\bm u}[k] + E {\bm d}[k] + {\bm w}[k]\\
{\bm y}[k]    &= C {\bm x}[k] + D {\bm u}[k] + F {\bm d}[k] + {\bm v}[k]. 
\end{cases}
\label{eq:meas}
\end{equation}
Here, ${\bm x}_k\in\mathbb{R}^n$, ${\bm y}_k\in\mathbb{R}^m$, $u_k\in\mathbb{R}^p$ are known, \mbox{${\bm w}[k] \sim \N({\bm 0}_n,Q)$}, \mbox{${\bm v}[k]\sim \N({\bm 0}_m,R)$}, while  ${\bm d}[k]\in\mathbb{R}^h$ is an unknown exogenous disturbance which satisfies ${\bm d}[k]\equiv {\bm 0}_h$ in the nominal case.

Moreover, let us consider a classical Kalman Filter estimation process in the form
\begin{equation}
    \label{eq:kalman}
    \begin{cases}
        \hat {\bm x}_{k+1|k} &= A \hat {\bm x}_{k|k} + B {\bm u}[k],\\
        P_{k+1|k} &= A P_{k|k} A^\top + Q,\\
        \hat{{\bm x}}_{k+1|k+1}&=\hat{{\bm x}}_{k+1|k}+K[k+1] {\bm e}[k+1],\\
        P_{k+1|k+1} &= (I-K[k+1] C)P_{k+1|k},
    \end{cases}
\end{equation}
where
$${\bm e}[k] \coloneqq {\bm y}[k] - C\hat {\bm x}_{k|k-1} - D {\bm u}[k],
$$
$$
S[k] \coloneqq C P_{k|k-1} C^\top + R, \quad  K[k] \coloneqq P_{k|k-1} C^\top S^{-1}[k].
$$
In order for the Kalman Filter to be able to reconstruct the state in nominal conditions, let us make the following standard assumption.
\begin{assumption}
\label{ass:detectable}
    The pair $(A,C)$ is detectable.
\end{assumption}

Before presenting the two problems at hand in this paper, let us provide the following definition.
\smallskip
\begin{definition}[Whitened Innovation]
The whitened innovation at time $k$ is given by
\mbox{$
{\bm z}[k] \coloneqq L^{-1}[k]\,{\bm e}[k] \in \mathbb{R}^m,
$} 
where $L[k]$ is the unique lower-triangular matrix (from the Cholesky factorization) such that $S[k]=L[k] L^\top[k]$ is the innovation covariance.
\end{definition}
\smallskip
As noted in the following remark, the above normalization has the effect of standardizing\footnote{Notice that $S[k]$ is the predicted innovation covariance from the filter, 
so whitening is only approximate; any model or noise mismatch will cause residual 
covariance deviations that can be revealed through statistical testing.} the statistical properties of the innovations.
\smallskip
\begin{remark}
\label{rem:white}
Premultiplication by $L^{-1}[k]$ standardizes the innovations (up to modeling/tuning errors). Indeed, if
$\mathbb{E}[{\bm e}[k]] = 0$ and $\mathrm{Cov}({\bm e}[k]) = S[k]$, then 
\mbox{$\mathrm{Cov}({\bm z}[k]) = L^{-1}[k]\, S[k]\, (L^{-1}[k])^\top= I_m. $}
Hence, ${\bm z}[k]$ has identity covariance. Moreover, under
\emph{correct tuning}\footnote{By correct tuning, it is meant that the Kalman filter is designed with the true system model and the correct noise statistics $Q,R$.}, the
innovations ${\bm e}[k]$ have zero mean, so that $\mathbb{E}[{\bm z}[k]]
\approx 0$. Consequently, the sequence $\{{\bm z}[k]\}$ is (approximately)
i.i.d.\ $\mathcal{N}({\bm 0}_m,I_m)$.
\end{remark}
\smallskip
On this basis, the central statistical question can be formulated as
the following hypothesis test.
\smallskip
\begin{hypothesistesting}\,
\noindent $\texttt{H}_0:$ $\{{\bm z}[k]\}$ are i.i.d.\ $\mathcal{N}({\bm 0}_m,I_m)$.\\
\noindent $\texttt{H}_1:$ any deviation from $\texttt{H}_0$.
\end{hypothesistesting}
\smallskip

It is now possible to introduce the two problems of interest.
In particular, consider two complementary problems: (i) from the observer's perspective, testing whether the system operates under the nominal hypothesis; and (ii) from the adversary's perspective, designing an injected disturbance ${\bm d}[k]$ that maximizes impact while evading detection.
\smallskip
\begin{problem}[Hypothesis testing]
\label{prob:hypothesistesting}
Decide between $\texttt{H}_0$ and $\texttt{H}_1$ based on the whitened innovations $\{{\bm z}[k]\}_{k=0}^T$.
\end{problem}
\smallskip
\begin{problem}[Adversarial disturbance design]
\label{prob:adversarial}
Find a disturbance sequence $\{{\bm d}[k]\}_{k=0}^T$
that inflicts maximum damage on the system while remaining undetected
by the hypothesis test.
\end{problem}
\smallskip
Problem~\ref{prob:adversarial} is left intentionally vague, i.e., definitions of ``damage" and ``undetected" depend on the SDS tests. Therefore, first Problem~\ref{prob:hypothesistesting} is solved and then the attacker problem is refiend.

\section{Statistical Detection Suite}
\label{sec:ids}
In order to solve Problem~\ref{prob:hypothesistesting}, an SDS is introduced. The SDS consists of four complementary checks, each
targeting one defining property of $\texttt{H}_0$: (a) Zero-mean; (b) unitary covariance; (c) Gaussianity; (d)  temporal independence.
Each test yields a $p$-value quantifying the evidence against
$\texttt{H}_0$. 
In the remainder of this section, the construction of each test is detailed.

\noindent {\bf Zero{\textendash}mean (bias) test.} A first property of the whitened innovations under $\texttt{H}_0$ is that
their mean must vanish. To check this,
the empirical average 
\[
\bar{{\bm z}}=\frac{1}{T+1}\sum_{k=0}^T {\bm z}[k]
\]
is computed.
Since the covariance of ${\bm z}[k]$ is known to be $I_m$, then the
quadratic form
\begin{equation}
\texttt{T}_\mu \coloneqq (T+1)\,\bar{{\bm z}}^\top \bar{{\bm z}}
\end{equation}
approximately follows a $\chi^2$ distribution with $m$ degrees of freedom under
$\texttt{H}_0$~\citep{Hotelling1931}, assuming the whitened innovations are nearly independent. In practice, one computes $\texttt{T}_\mu$
from the data, and compares it with the quantiles of $\chi^2_m$: if
$T_\mu$ falls in the upper tail (e.g., above the $95$th percentile), the
null hypothesis of zero mean is rejected and a bias is declared.

\noindent{\bf Covariance consistency (NIS).} 
A second property of the whitened innovations under $\texttt{H}_0$ is
that their covariance must equal the identity matrix. 
A widely used
statistic to check this is the \emph{Normalized Innovation Squared}
(NIS)~\citep{Bierman}. For each sample ${\bm z}[k]\in\mathbb R^m$, the per–sample NIS is
$q_k=\|{\bm z}[k]\|_2^2$. Summing over the horizon $\{0,\dots,T\}$ yields the
aggregate statistic
\[
\texttt{Q} \;\coloneqq\; \sum_{k=0}^T q_k \;=\; \sum_{k=0}^T \|{\bm z}[k]\|_2^2.
\]
Under $\texttt{H}_0$, the variables ${\bm z}[k]$ are i.i.d.\ $\mathcal N(0,I_m)$,
so $\texttt{Q}$ is $\chi^2$ distributed with $\nu=m(T+1)$ degrees of freedom\footnote{Although the innovations are normalized using the filter's predicted covariance $S[k]$, any mismatch between the predicted and true covariance may cause residual deviations, which the NIS statistic is designed to capture.
}. For a given significance level $\alpha$, the test threshold is
\mbox{$\tau_\alpha \;\coloneqq\; \chi^2_{\nu}(1-\alpha)$}, and the test rejects $\texttt{H}_0$ if $\texttt{Q}>\tau_\alpha$.
\color{black}
For later use, note that under $\texttt{H}_0$ it holds
$\mathbb{E}[\texttt{Q}]=m(T+1)$.
It is therefore convenient to introduce the excess-NIS budget
\[
\varepsilon_{\rm cov}
\;\coloneqq\;
\tau_\alpha - m(T+1),
\]
which represents the maximum admissible increase in the expected NIS
before the detector rejects $\texttt{H}_0$ at significance level $\alpha$.
\color{black}

{\bf Gaussianity (shape).} A third defining property of the whitened innovations under $\texttt{H}_0$ is that they are normally distributed. To capture dependencies across components, multivariate tests are preferable. A widely used choice is Mardia's test~\citep{Mardia1970}.
Specifically, for each $k$ let the (centered) Mahalanobis distance $\delta_k$ be such that
\[
\delta_k^2=\big({\bm z}[k]-\bar{\bm z}\big)^\top \hat\Sigma^{-1}\big({\bm z}[k]-\bar{\bm z}\big),
\]
where $\hat\Sigma$ is the sample covariance\footnote{\color{black}
It is assumed that $T+1>m$ so that $\hat\Sigma$ is invertible; otherwise one may use a regularized inverse.
\color{black}
} of the whitened
innovations, i.e., 
\mbox{$\hat\Sigma=\frac{1}{T+1}\sum_{k=0}^T ({\bm z}[k]-\bar{{\bm z}})({\bm z}[k]-\bar{{\bm z}})^\top.$}
Then, define Mardia's multivariate kurtosis as
\[
b_{2,m}=\frac{1}{T+1}\sum_{k=0}^{T}\delta_k^4.
\]
Under $\texttt{H}_0$ (multivariate normality), the statistic $b_{2,m}$
is centered around the theoretical value $\mu_2 = m(m+2)$, with sampling variability that decreases as the number of observations
increases. More precisely, its variance is
\[
\sigma_2^2 = \frac{8m(m+2)}{T+1}.
\]
It is therefore natural to form a standardized statistic
\[
\texttt{Z}_{\rm kurt}=\frac{b_{2,m}-\mu_2}{\sigma_2},
\]
which measures how many standard deviations the observed kurtosis lies
away from its expected value. By the central limit theorem,
$\texttt{Z}_{\rm kurt}$ is approximately standard normal,
\mbox{$\texttt{Z}_{\rm kurt}\ \ \overset{\text{approx}}{\sim}\ \ \mathcal N(0,1)$}.

Given an observed value of $\texttt{Z}_{\rm kurt}$, the corresponding
two-sided $p$-value is \mbox{$p_{\rm norm}=2\big(1-\Phi(|\texttt{Z}_{\rm kurt}|)\big)$}, where $\Phi(\cdot)$ denotes the cumulative distribution function (CDF)
of a standard normal random variable. Intuitively, $p_{\rm norm}$
quantifies the probability of observing a deviation at least as extreme
as $|\texttt{Z}_{\rm kurt}|$ under Gaussianity. A small value of $p_{\rm norm}$
(e.g.\ below $0.05$) provides statistical evidence against normality\footnote{Mardia's framework also includes a skewness statistic. Since this paper aims to
retain one representative check per property, the kurtosis-based
test is adopted; however, both statistics could be reported if a more exhaustive
normality check is desired.}.

{\bf Whiteness (temporal independence).} A final property under $\texttt{H}_0$ is temporal independence: the sequence $\{{\bm z}[k]\}_{k=0}^T$ should be uncorrelated across time. To assess this, the multivariate Ljung--Box/Hosking portmanteau statistic~\citep{Hosking1980} is adopted.

To this end, consider lags $\tau=1,\dots,L$ (with $L\ll T$), and define the sample autocovariances $\widehat R(\tau)$ as 
\[
\widehat R(\tau)=\frac{1}{T+1}\sum_{k=\tau}^{T}
\big({\bm z}[k]-\bar{\bm z}\big)\big({\bm z}[k-\tau]-\bar{\bm z}\big)^\top.
\]
Moreover define the  portmanteau statistic\footnote{
Hosking~\citep{Hosking1980} used a normalization factor $n^2/(n-r)$,
whereas Ljung and Box~\citep{LjungBox1978} proposed the univariate correction
$n(n+2)/(n-k)$.  Replacing $(T{+}1)^2$ with $(T{+}1)(T{+}3)$ in~\eqref{eq:portmanteau}
is a practical small--sample adaptation, asymptotically equivalent to Hosking's form.} as 
\begin{equation}
\texttt{W}_L\!=\!(T+1)(T+3)\sum_{\tau=1}^{L}\frac{1}{T+1-\tau}\,
\mathrm{tr}\big(\widehat R(\tau)^\top \widehat R(\tau)\big).
\label{eq:portmanteau}
\end{equation}
Under $\texttt{H}_0$, $\texttt{W}_L$ is commonly approximated by
$\chi^2_{m^2 L}$. \textcolor{black}{Hence,} the whiteness $p$-value is
$p_{\rm white}=\Prob(\chi^2_{m^2L}\ge \texttt{W}_L)$.

\begin{remark}
If desired, the four $p$-values obtained from the mean, covariance,
Gaussianity, and whiteness tests can be combined into a single decision
statistic, for instance through Fisher's method~\citep{Fisher1925}.
\end{remark}

\section{Attacker Problem with Filter Design}
\label{sec:attacker_filter}

\textcolor{black}{The following attack model is introduced as a representative construction of stealthy disturbances that preserve innovation statistics. While a specific parametrization is adopted for analytical tractability, the detection methodology developed in this work is not restricted to this particular model, but applies more generally to attacks that aim to remain hidden within the innovation process.}

\noindent {\bf Attack Model.} an attacker is considered that leverages a linear time–invariant (LTI) FIR filter with coefficients encoded by the matrices $M[0],\ldots, M[r]\in\mathbb{R}^{h\times s}$ to shape a zero–mean Gaussian seed ${\bm{\xi}[k]}\in\mathbb{R}^s$, assumed independent of process and measurement noise, into the disturbance sequence.
This
choice is deliberate: linear filtering of a zero--mean Gaussian seed
preserves both zero mean and joint Gaussianity, so the SDS mean (bias)
and Gaussianity (shape) tests are neutralized by construction. At the
same time the filter injects controlled covariance and temporal
correlation (coloring) into ${\bm d}$, which the attacker can exploit to
maximize expected damage while attempting to remain below the SDS
covariance (NIS) and whiteness detection thresholds.
As a result, a scenario is considered where the attacker injects a signal in the form
\[
{\bm d}[k] = \sum_{\tau=0}^{r} M[\tau]\,\bm{\xi}[k-\tau],
\qquad \bm{\xi}[k]\sim \mathcal N({\bm 0}_s,I_s),
\]
or, in stacked form, ${\bm d}_{0:T}=\mathcal T_M\,\bm{\xi}_{0:T}$, where
$\mathcal T_M$ is the block–Toeplitz (convolution) matrix induced by the
filter matrices $\{M[\tau]\}$, ${\bm d}_{0:T}$ is the stack of the samples ${\bm d}[k]$ in $\{0,\ldots,T\}$, i.e., ${\bm d}_{0:T}=[{\bm d}^\top[0],\ldots,{\bm d}^\top[T]]^\top$ and, likewise  ${\bm \xi}_{0:T}$ is the stack of the terms ${\bm \xi}[k]$.
The block–Toeplitz convolution matrix $\mathcal T_M$ is defined
entrywise by its $h\times s$ blocks: for $i,j\in\{0,\dots,T\}$ the $(i,j)$th
block is
\[
(\mathcal T_M)_{ij} =
\begin{cases}
M[i-j], & 0\le i-j \le r,\\[4pt]
0, & \text{otherwise},
\end{cases}
\]
so that $\mathcal T_M\in\mathbb R^{h(T+1)\times s(T+1)}$ and
$d_{0:T}=\mathcal T_M\xi_{0:T}$.

Notably, as a result of the attack, the induced innovation deviation is
\mbox{$\Delta{\bm z}_{0:T} \;=\; G_T\,{\bm d}_{0:T}
\;=\; G_T\,\mathcal T_M\,\bm{\xi}_{0:T},$}, where $G_T$ is the finite–horizon convolution matrix that maps the stacked
disturbance ${\bm d}_{0:T}$ into $\Delta{\bm z}_{0:T}$, the latter being the stack of the deviations $\Delta {\bm z}[k]$ between the nominal ${\bm z}[k]$ whitened innovations and ${\bm z}^a[k]$, i.e., those obtained as a result of the attack.
Hence, $G_T$ represents the \emph{detection channel}, since it captures
how the attack is observed by the SDS. 
The disturbance also propagates through the system dynamics, producing
an attack–induced deviation of the state trajectory
\mbox{$\Delta {\bm x}_{0:T} \;=\; H_T\,{\bm d}_{0:T}
\;=\; H_T\,\mathcal T_M\,\bm{\xi}_{0:T}.$}
Here $H_T$ is the finite–horizon convolution matrix from disturbances to
states. Note that $\Delta{\bm x}_{0:T}$ represents the component of the
state trajectory due solely to the attacker; in nominal conditions
(${\bm d}\equiv 0$) it should hold $\Delta{\bm x}_{0:T}=0$, although the
state still fluctuates due to the process noise ${\bm w}[k]$.

\noindent {\bf Problem Formulation.} Notice that, as mentioned above, the attack model \textcolor{black}{ensures} that mean and Gaussianity test are passed. 

Let us now consider the covariance test and let  
$$\Sigma_d \coloneqq \mathbb{E}[{\bm d}_{0:T}{\bm d}_{0:T}^\top]
= \mathcal T_M \underbrace{\mathbb{E}[{\bm \xi}_{0:T}{\bm \xi}_{0:T}^\top]}_{I_{s(T+1)}}\mathcal T_M^\top= \mathcal T_M\mathcal T_M^\top.$$ 
The induced covariance of innovation
deviations is 
$$\Sigma_{\Delta z}= \mathbb{E}[\Delta{\bm z}_{0:T}\Delta{\bm z}_{0:T}^\top]= 
 G_T \Sigma_d G_T^\top.$$ 
Let us now consider the NIS statistic. Assuming the attack seed is independent of the nominal process and
measurement noises, the cross--covariance
$\mathbb{E}[\Delta{\bm z}\,{\bm z}_0^\top]$ vanishes. Without attack,
therefore, it holds
\mbox{$\mathbb{E}[\texttt{Q}]_{\rm nom} \;=\; m(T+1)$}, while under attack (assumed independent of nominal noises) the expected NIS is shifted by the trace contribution\footnote{The cross term with the nominal innovations vanishes by independence.}
\[
\mathbb{E}[\texttt{Q}] \;=\; m(T+1) \;+\; \mathrm{tr}\big(Q_G\,\Sigma_d\big), \quad Q_G=G_T^\top G_T.
\]
To keep the false alarm rate close to the nominal level, the
constraint \mbox{$\mathrm{tr}\!\big(Q_G\,\Sigma_d\big) \;\le\; \varepsilon_{\rm cov}$} is enforced.

Beyond mean, covariance, and Gaussianity, the SDS also checks the
\emph{temporal independence} of the innovation sequence via Ljung--Box portmanteau statistic.
In particular, let $\mathcal L_\tau(\cdot)$ denote the linear operator that
extracts from a block covariance matrix the empirical autocovariance at
lag $\tau$. Since under attack the stacked innovation deviation has
covariance $\Sigma_{\Delta z}=G_T \Sigma_d G_T^\top$, the contribution of lag $\tau$ to the portmanteau statistic is
$\|\mathcal L_\tau(\Sigma_{\Delta z})\|_F^2$. Summing these contributions
with appropriate weights $w_\tau$ yields the whiteness test constraint
\begin{equation}
\label{eq:white_bound}
\sum_{\tau=1}^{L} w_\tau \,
\big\| \mathcal L_\tau\!\big(\Sigma_{\Delta z}\big) \big\|_F^2
\;\le\; \varepsilon_{\rm white},
\qquad
w_\tau = \frac{(T+1)(T+3)}{(T+1)-\tau}.
\end{equation}
Here $\varepsilon_{\rm white}$ is the significance threshold corresponding to a
$\chi^2$ distribution with $m^2L$ degrees of freedom. The weights
$w_\tau$ correct for finite-sample effects so that, under the null
hypothesis of whiteness, the statistic is asymptotically $\chi^2$
distributed.

\color{black}
\begin{remark}
From the attacker's perspective, this condition complements a simple covariance test:
while a covariance-based detector constrains the overall variance, the portmanteau
whiteness test further limits admissible temporal correlation patterns, requiring the
attacker to operate within a bounded correlation budget.
\end{remark}

\color{black}

Finally, to avoid trivial unbounded attacks, 
we limit the total disturbance energy over the design horizon by enforcing
\mbox{$\mathrm{tr}(\Sigma_d) 
= \mathrm{tr}(\mathcal{T}_M \mathcal{T}_M^\top) 
= \|\mathcal{T}_M\|_F^2 
\le D,$} with $D>0$.
The attacker's objective is now characterized. Specifically, a scenario is considered where the attacker maximizes the \emph{expected damage}  transmitted through the
plant's damage channel. With $Q_H\coloneqq H_T^\top H_T$, using the well known property 
\mbox{$\mathbb{E}[{\bm c}^{\top}Y{\bm c}]=\mathrm{tr}(Y\,\mathbb{E}[{\bm c}^{\top}{\bm c}]),$} 
we have that
\[
\begin{aligned}
  \mathbb{E}[\|H_T\,{\bm d}_{0:T}\|^2]
\;&=\; \mathbb{E}\!\big[{\bm d}_{0:T}^\top Q_H {\bm d}_{0:T}\big]
\;=\; \mathrm{tr}\!\big(Q_H\,\Sigma_d\big)\\
\;&=\; \mathrm{tr}\!\big(Q_H\,\mathcal T_M\mathcal T_M^\top\big).  
\end{aligned}
\]
Overall, the resulting filter–design attacker problem is given by 
\begin{align}
\label{eq:attacker_M}
\max_{\,\{M[\tau]\}_{\tau=0}^{r}}~&
\mathrm{tr}\!\big(Q_H\,\mathcal T_M\mathcal T_M^\top\big)\\
\text{s.t.}\quad
& \mathrm{tr}\!\big(Q_G\,\mathcal T_M\mathcal T_M^\top\big) \;\le\; \varepsilon_{\rm cov}, \nonumber\\
& \sum_{\tau=1}^{L} w_\tau \,
\big\| \mathcal L_\tau\!\big(G_T\mathcal T_M\mathcal T_M^\top G_T^\top\big) \big\|_F^2
\;\le\; \varepsilon_{\rm white}, \nonumber\\
& \|\mathcal T_M\|_F^2 \;\le\; D. \nonumber\\
& \mathcal{T}_M\in \mathbb{T}_M, \nonumber
\end{align}
where $\mathbb{T}_M$ is the set of block–Toeplitz convolution matrices induced by the FIR
coefficients $\{M[\tau]\}_{\tau=0}^{r}$.
\begin{remark}
The Toeplitz constraint $\mathcal T_M\in\mathbb{T}_M$ encodes linear equalities while the
NIS and energy constraints are convex quadratics in the coefficients, and the whiteness
constraint is quartic in the coefficients $M[\cdot]$ through $G_T\mathcal T_M\mathcal T_M^\top G_T^\top$,
hence \eqref{eq:attacker_M} is nonconvex in general.
\end{remark}

\subsection{Convex relaxation}
While the attacker’s filter-design problem is inherently nonconvex due to the
whiteness constraint, a tractable convex relaxation can be obtained.
In particular, replacing the structured disturbance covariance 
$\Sigma_d = \mathcal T_M \mathcal T_M^\top$ by
a generic $\Sigma\succeq 0$ yields the convex program
\begin{align}
\label{eq:attacker_Sigma}
\max_{\Sigma\succeq 0}\quad
& \langle Q_H,\Sigma\rangle \\
\text{s.t.}\quad
& \langle Q_G,\Sigma\rangle \;\le\; \varepsilon_{\rm cov}, \nonumber\\
& \sum_{\tau=1}^{L} w_\tau \,
\big\| \mathcal L_\tau\!\big(G_T \Sigma G_T^\top\big) \big\|_F^2
\;\le\; \varepsilon_{\rm white}, \nonumber\\
& \mathrm{tr}(\Sigma) \;\le\; D. \nonumber
\end{align}
Each term is either linear in $\Sigma$ or a squared Frobenius norm of a
linear map of $\Sigma$, hence convex. Problem~\eqref{eq:attacker_Sigma}
therefore provides a tractable \emph{convex relaxation} of the attacker
design, and its optimal value is an \emph{upper bound} on the maximal damage achievable by any FIR filter of any order, since
realizability of $\Sigma=\mathcal T_M\mathcal T_M^\top$ is not enforced.

\begin{remark}
Problem~\eqref{eq:attacker_Sigma} is always feasible
(e.g.\ $\Sigma=0$ satisfies all constraints).
Moreover, the feasible set is bounded by $\mathrm{tr}(\Sigma)\le D$,
and closed under the constraint $\Sigma\succeq 0$. Hence an optimal
solution $\Sigma^\star$ exists. 
\end{remark}

\noindent {\bf Existence of nontrivial solutions}
Let us now establish feasibility conditions for the
convex relaxation in Eq.~\eqref{eq:attacker_Sigma}, clarifying when a nontrivial attacker covariance $\Sigma\neq 0$ exists and how its optimal value depends on the SDS
thresholds and the detection channel.
\smallskip
\begin{theorem}
\label{thm:nontrivial-existence}
Consider the relaxation \eqref{eq:attacker_Sigma}. Then, the following statements hold true.
\begin{enumerate}
\item[\emph{(i)}] If at least one of the SDS thresholds is strictly positive
($\varepsilon_{\rm cov}>0$ or $\varepsilon_{\rm white}>0$), then there exists a
nonzero feasible covariance $\Sigma\neq 0$. Moreover, if
$Q_H\not\equiv 0$, \color{black}
the covariance
$\Sigma^\dag=\mathcal{T}_M^\dag \mathcal{T}_M^{\dag\top}$
corresponding to an optimal solution $\mathcal{T}_M^\dag$ of the original problem
in~\eqref{eq:attacker_M} is also nonzero.
\color{black}

\item[\emph{(ii)}] If both thresholds are zero ($\varepsilon_{\rm cov}=0$ and $\varepsilon_{\rm white}=0$),
then there exists a feasible nonzero $\Sigma$ if and only if
$\ker(G_T)\neq \{0\}$. In that case the optimal value is
\[
\max_{\Sigma\succeq 0,\ \mathrm{tr}(\Sigma)\le D,\ G_T\Sigma G_T^\top=0}
\ \langle Q_H,\Sigma\rangle
\;=\; D\,\lambda_{\max}(P Q_H P),
\]
where $P$ is the orthogonal projector onto $\ker(G_T)$.
\end{enumerate}
\end{theorem}

\begin{proof}
(i) Suppose $\varepsilon_{\rm cov}>0$ or $\varepsilon_{\rm white}>0$.  
Choose any vector ${\bm u}$ such that ${\bm u}^\top Q_H {\bm u}>0$, and define
$\Sigma = t\,{\bm u}{\bm u}^\top$ with $t>0$.  
Then
\[
\langle Q_H,\Sigma\rangle =t\,\mathrm{tr}(Q_H \,{\bm u}{\bm u}^\top)= t\, \mathrm{tr}({\bm u}^\top Q_H \,{\bm u})= t\,{\bm u}^\top Q_H {\bm u} > 0.
\]
The covariance constraint
$\langle Q_G,\Sigma\rangle = t\,{\bm u}^\top Q_G {\bm u}$ scales linearly in $t$,  
while the constraint
\mbox{$\sum_\tau w_\tau \|\mathcal L_\tau(G_T \Sigma G_T^\top)\|_F^2$}
scales quadratically in $t$.  
For sufficiently small $t>0$, both constraints are satisfied, so
$\Sigma$ is feasible and strictly improves the objective with respect to $\Sigma=0$.  
Therefore $\Sigma^\dag\neq 0$ whenever $Q_H\not\equiv 0$.

(ii) Suppose $\varepsilon_{\rm cov}=0$ and $\varepsilon_{\rm white}=0$.  
If $\Sigma$ is feasible, then necessarily
$G_T \Sigma G_T^\top=0$, which means
$\mathrm{range}(\Sigma)\subseteq \ker(G_T)$.  
Conversely, any $\Sigma\succeq 0$ with
$\mathrm{range}(\Sigma)\subseteq \ker(G_T)$ automatically satisfies the
constraints. Thus the feasible set reduces to
\mbox{$\{\Sigma\succeq 0:\ \mathrm{tr}(\Sigma)\le D,\ \mathrm{range}(\Sigma)\subseteq\ker(G_T)\}.$}

Maximizing $\langle Q_H,\Sigma\rangle$ over this set is achieved by a
rank-one solution $\Sigma = D vv^\top$, where $v$ is the unit-norm dominant
eigenvector of $P Q_H P$. 
Let $\Sigma = D\, v v^\top$ with $v\in\ker(G_T)$ and $\|v\|_2=1$. Then
$\mathrm{tr}(\Sigma)=D\,\mathrm{tr}(vv^\top)=D$, so $\Sigma$ is feasible and
$\mathrm{range}(\Sigma)\subseteq\ker(G_T)$. The objective evaluates to
\mbox{$\langle Q_H,\Sigma\rangle \;=\; \mathrm{tr}(Q_H\, D v v^\top)
\;=\; D\, v^\top Q_H v$}.
Because $v\in\ker(G_T)$, it holds $Pv=v$ for the orthogonal projector
$P$ onto $\ker(G_T)$, hence \mbox{$v^\top Q_H v \;=\; v^\top P Q_H P v$}.
Maximizing $\langle Q_H,\Sigma\rangle$ over all feasible rank-one choices is
therefore equivalent to maximizing the Rayleigh quotient
$R(v)=\tfrac{v^\top (P Q_H P) v}{\|v\|_2^2}$ over $v\neq 0$, whose maximum is
$\lambda_{\max}(P Q_H P)$, achieved at any unit eigenvector $v$ of $P Q_H P$
associated with $\lambda_{\max}$. Consequently,
\[
\max_{\substack{\Sigma\succeq 0,\,\mathrm{tr}(\Sigma)\le D,\\
\mathrm{range}(\Sigma)\subseteq\ker(G_T)}} \langle Q_H,\Sigma\rangle
\;=\; D\,\lambda_{\max}(P Q_H P).
\]
The proof is complete.
\end{proof}

\begin{remark}
The kernel $\ker(G_T)$ identifies disturbance directions that remain
invisible to the innovations over the horizon, i.e., finite-horizon
undetectable subspaces. If this subspace is nontrivial and there exists
$v \in \ker(G_T)$ such that $v^\top Q_H v > 0$, then the attacker can
inject disturbances that are \emph{perfectly stealthy} while still
producing nonzero state deviation. If
$\ker(G_T)=\{0\}$, then every nonzero disturbance affects the
innovations, so nonzero feasible covariances require positive thresholds
$\varepsilon_{\rm cov}$ or $\varepsilon_{\rm white}$.
\end{remark}

\color{black}

\noindent\textbf{Recovering an implementable filter.}
The solution of the convex relaxation, $\Sigma^\star$, is in general an arbitrary covariance
matrix, whereas realizability requires $\Sigma_d=\mathcal T_M\mathcal T_M^\top$ with
$\mathcal T_M$ block--Toeplitz, i.e., induced by a finite--impulse--response (FIR) filter driven
by white noise. To obtain a feasible approximation, it is possible to proceed as follows.

First, a low--rank factorization of $\Sigma^\star$ is computed via the eigen--decomposition
$\Sigma^\star=V\Lambda V^\top$ and retain the dominant $k$ modes,
$\widetilde{\mathcal T}_0 = V_k\Lambda_k^{1/2}$, which provides the best rank--$k$
approximation in Frobenius norm (Eckart--Young theorem).

The rank-$k$ factor $\widetilde{\mathcal T}_0$ obtained in the previous step is not unique:
for any orthogonal matrix $U$, the rotated factor $\widetilde{\mathcal T}_0 U$ induces the
same covariance
$\Sigma^\star \approx \widetilde{\mathcal T}_0\widetilde{\mathcal T}_0^\top$.
As a consequence, projecting a specific factor realization onto the Toeplitz set
$\mathbb T_M$ may yield substantially different recovered covariances, even though they
all correspond to the same relaxed solution $\Sigma^\star$.
This \emph{factor-rotation ambiguity} is critical in the present setting, since both the
attacker objective and all SDS constraints depend on the disturbance only through its
induced covariance.

To mitigate this ambiguity and \textcolor{black}{better preserve} the covariance-level structure imposed by
the convex relaxation, FIR realizability is enforced directly at the covariance level by
fitting a block-Toeplitz convolution operator $\widehat{\mathcal T}$ such that $\widehat{\mathcal T}\widehat{\mathcal T}^\top$ best matches
$\Sigma^\star$.
Specifically, the FIR structure is recovered by solving
\[
\widehat{\mathcal T}
=\arg\min_{\mathcal T\in\mathbb T_M}
\|\mathcal T\mathcal T^\top-\Sigma^\star\|_F^2
\;+\;\zeta\|Q_G^{1/2}\mathcal T\|_F^2,\,\text{ s.t. } \|\mathcal T\|_F^2\le D,
\]
where the second term penalizes energy injected along the detection channel and improves
numerical robustness.
Although the problem is nonconvex, it is smooth and low-dimensional, being parameterized
by the FIR coefficients $\{M[\tau]\}_{\tau=0}^{r}$, and can be solved efficiently in
practice using standard gradient-based methods with suitable initialization.
$\widehat\Sigma=\widehat{\mathcal T}\widehat{\mathcal T}^\top$ is, by construction,
aligned with the relaxed covariance $\Sigma^\star$.

To ensure compliance with all detector constraints, the recovered covariance
$\widehat\Sigma$ is subsequently subjected to a uniform amplitude scaling step,
described in detail at the end of this section.
\color{black}

To quantify \textcolor{black}{ the potential performance loss when moving from the convex relaxation to an implementable FIR disturbance, we derive}  a worst--case bound on the gap between the relaxed
objective and the cost attained after structured recovery.

\begin{theorem}
Let $\Sigma^\star$ be an optimal solution to~\eqref{eq:attacker_Sigma} and define
$J_{\rm relax}=\langle Q_H,\Sigma^\star\rangle$.
Let $\widehat{\mathcal T}\in\mathbb T_M$ be any recovered block-Toeplitz convolution operator
(e.g., obtained by the covariance-fitting step), and define the corresponding recovered covariance and cost as
$\widehat\Sigma=\widehat{\mathcal T}\widehat{\mathcal T}^\top$ and
$J_{\rm rec}=\langle Q_H,\widehat\Sigma\rangle$.
Then
\begin{equation}
\label{eq:relax1}
\left|J_{\rm relax}-J_{\rm rec}\right|
\ \le\
\|Q_H\|_2\,\|\Sigma^\star-\widehat{\Sigma}\|_* .
\end{equation}
\end{theorem}

\begin{proof}
By H\"older's inequality for the spectral/nuclear norm pair,
$
\big|\langle Q_H,\widehat{\Sigma}-\Sigma^\star\rangle\big|
\le
\|Q_H\|_2\,\|\widehat{\Sigma}-\Sigma^\star\|_*,
$
which yields~\eqref{eq:relax1}.
\end{proof}
\color{black}

\begin{remark}
The bound in~\eqref{eq:relax1} provides an \emph{a posteriori} certificate, in the sense
that it can be evaluated numerically for the \emph{unscaled} recovered covariance
$\widehat\Sigma$ once the relaxation and structured recovery steps have been performed.
Importantly, this bound is agnostic to the specific recovery procedure used to obtain
$\widehat{\mathcal T}$ and holds for any block-Toeplitz FIR realization.

Since the structured recovery step is carried out via a nonconvex covariance-fitting
procedure, no closed-form \emph{a priori} bound on the relaxation--recovery gap in terms
of the unknown Toeplitz optimum is available in general.
For this reason, performance guarantees are provided through explicit numerical
certification and subsequent amplitude scaling, which together ensure feasibility with
respect to all detector constraints.
\end{remark}

The bounds derived above are inherently conservative and do not, by themselves,
guarantee satisfaction of the detector constraints after structured recovery.
Accordingly, additional corrective steps may be required.
In particular, if needed, positive semidefiniteness is restored by spectral clipping,
and the recovered FIR operator $\widehat{\mathcal T}$ is subsequently subjected to a
uniform amplitude scaling by $\gamma\in(0,1]$, which induces a corresponding scaling of
the covariance by $\gamma^2$ and enforces all detector constraints while preserving the
FIR temporal structure.

Notably, related approaches based on rank--regularized least--squares are commonly used to
enforce structure directly on system or filter coefficients in identification problems;
see, e.g.,~\cite{noom2024simultaneously,wang2025continuous}.
In the present setting, however, the attacker objective and all SDS constraints depend on
the FIR taps only through the induced disturbance covariance
$\Sigma_d=\mathcal T_M\mathcal T_M^\top$, and in particular the portmanteau (whiteness)
constraint is convex in $\Sigma_d$ but quartic in the taps.
For this reason, optimization is carried out over $\Sigma_d$ in the convex relaxation and enforce FIR
realizability through a subsequent structured factorization and covariance-consistent
recovery, followed by amplitude scaling, rather than via rank constraints imposed directly
on the filter coefficients.
The existence of such a scaling is guaranteed by the following result.
\begin{proposition}
\label{prop:gamma_specific}
Let $\varepsilon=\{\varepsilon[k]\}$ denote the (attack-induced) innovation component produced by the
recovered FIR disturbance model, and define the scaled sequence $\varepsilon_\gamma=\{\gamma\,\varepsilon[k]\}$, with $\gamma\ge 0$.
Let $T_{\mathrm{NIS}}(\cdot)$ be the covariance/NIS statistic and $T_{\mathrm{P}}(\cdot)$ the
portmanteau statistic (computed up to lag $L$) used in the SDS, with thresholds
$\varepsilon_{\rm cov}$ and $\varepsilon_{\rm white}$.
Then, \mbox{$T_{\mathrm{NIS}}(\varepsilon_\gamma)=\gamma^2\,T_{\mathrm{NIS}}(\varepsilon)$}, \mbox{$T_{\mathrm{P}}(\varepsilon_\gamma)=\gamma^4\,T_{\mathrm{P}}(\varepsilon)$}, and 
\begin{equation}
    \label{eq:scaling}
    \gamma \le \bar\gamma \;\triangleq\;
\min\!\left\{
\sqrt{\frac{\varepsilon_{\rm cov}}{T_{\mathrm{NIS}}(\varepsilon)}},
\left(\frac{\varepsilon_{\rm white}}{T_{\mathrm{P}}(\varepsilon)}\right)^{\!1/4}
\right\}
\end{equation}
guarantees\footnote{\textcolor{black}{The convention adopted is that the corresponding bound is $+\infty$ whenever
$T_{\mathrm{NIS}}(\varepsilon)=0$ or $T_{\mathrm{P}}(\varepsilon)=0$.}} $T_{\mathrm{NIS}}(\varepsilon_\gamma)\le \varepsilon_{\rm cov}$ and
$T_{\mathrm{P}}(\varepsilon_\gamma)\le \varepsilon_{\rm white}$.
\end{proposition}
\begin{proof}
Since $\varepsilon_\gamma=\gamma\varepsilon$, second-order statistics
scale quadratically, yielding $T_{\mathrm{NIS}}(\varepsilon_\gamma)=\gamma^2T_{\mathrm{NIS}}(\varepsilon)$.
Moreover, the autocovariances satisfy $\widehat R_\gamma[\ell]=\gamma^2\widehat R[\ell]$ for each lag
$\ell$, hence any portmanteau statistic based on sums of squared autocovariances scales as
$\gamma^4$, i.e., $T_{\mathrm{P}}(\varepsilon_\gamma)=\gamma^4T_{\mathrm{P}}(\varepsilon)$.
Substituting into the inequalities and solving for $\gamma$ gives the stated bound.
\end{proof}

\color{black}

\section{Evaluation in Maritime Navigation Scenarios}
\label{sec:maritime}
A surface vessel tracked at AIS/GNSS rates is considered with sampling time
$T_s=5\,\mathrm{s}$. The nearly constant-velocity (CV) kinematic model~\citep{Singer1970,Fossen2011} has state
\mbox{$\bm x[k] \;=\; \begin{bmatrix} p_x[k] & p_y[k] & v_x[k] & v_y[k]\end{bmatrix}^\top$}, while the input $\bm u[k]\in\mathbb{R}^2$ amounts to surge/sway accelerations and the output corresponds to the positions, i.e., $\bm y[k]\in\mathbb{R}^2$.
The exogenous input $\bm d[k]\in\mathbb{R}^2$ (fault/attack) is \emph{state-level} and acts only on the
\emph{velocity} components, while there is \emph{no measurement injection} (i.e.,$F=0$).
Thus, the vessel's dynamic model features matrices
\[
A=\begin{bmatrix}
1 & 0 & T_s & 0\\
0 & 1 & 0 & T_s\\
0 & 0 & 1   & 0\\
0 & 0 & 0   & 1
\end{bmatrix},\,\,
B=\begin{bmatrix}
\frac{T_s^2}{2} & 0\\
0 & \frac{T_s^2}{2}\\
T_s & 0\\
0 & T_s
\end{bmatrix},\,\,
\begin{matrix}
  C=\begin{bmatrix}
1&0&0&0\\[2pt]0&1&0&0
\end{bmatrix}\\[1em]
E=\begin{bmatrix}
0&0&1&0\\[2pt]
0&0&0&1
\end{bmatrix}^\top.
\end{matrix}
\]
Process and measurement noises are zero-mean Gaussian,
$\bm w[k]\sim\mathcal N(\bm 0,Q)$, $\bm v[k]\sim\mathcal N(\bm 0,R)$.
The Kalman filter is tuned with diagonal covariances: 
the process noise $Q$ injects uncertainty only in the velocity channels, representing unmodelled accelerations via a velocity random walk, 
while the measurement noise $R$ acts only on the positions, modeling GNSS/AIS jitter. 
The chosen magnitudes are consistent with moderate sea state and consumer-grade receivers, i.e., 
\[
Q \;=\; \begin{bmatrix} 0_{2\times 2} & 0_{2\times 2}\\ 0_{2\times 2} & \sigma_{\!v}^2 I_2
\end{bmatrix},
\qquad
R \;=\; \sigma_{\!y}^2\, I_2,
\]
with \mbox{$\sigma_{\!v}=10^{-2}\,\mathrm{m/s}$} and
\mbox{$\sigma_{\!y}=5\,\mathrm{m}$}, which are compatible with commercial-grade marine navigation sensors. 
Notably, $(A,C)$ is observable.

With respect to the above model, two representative navigation contexts are considered: (i) a straight transit segment, and; (ii) a gentle maneuver segment in a traffic-regulated area. 
In this view, the input $\bm u[k]$ consists of a small surge acceleration with a short sway pulse to emulate a course adjustment. 

The SDS is configured as follows. Innovations are whitened via the Cholesky factor of the innovation covariance. Four tests are computed over the full trajectory at significance $\alpha=0.05$: (a) zero-mean, with statistic $T_\mu\sim\chi^2_{m}$ ($m{=}2$); (b) NIS, summing $\|z[k]\|_2^2$ with $\nu=m(T{+}1)$ degrees of freedom; (c) Mardia’s multivariate kurtosis (two-sided normal approximation for its standardized statistic); and (d) multivariate Ljung--Box whiteness with $L=10$ lags, using the $\chi^2_{m^2L}$ reference. 

Two disturbance patterns injected at the state level through $E$ are compared:
{\bf Attack A:} a colored first-order autoregressive disturbance (AR(1)) on velocities; {\bf Attack B:}  recovered FIR–Gaussian disturbance from the proposed relaxation and subsequent recovery approach.
\color{black}For Attack A, the injected disturbance acts on both velocity channels as a vector AR(1) process \mbox{$\mathbf d_k = \omega\,\mathbf d_{k-1} + \psi\,\boldsymbol{\xi}_k$} with \mbox{$\boldsymbol{\xi}_k \sim \mathcal N(\mathbf 0, I_2)$}, where  $\omega\in(0,1)$ is the coloring parameter and $\psi>0$ the innovation scale. 

\color{black}
A design window of $T_{\rm design}{+}1=31$ samples (i.e., $155$ seconds at $T_s=5$\,s) is adopted,
for which the nominal expectation of the NIS statistic is
$\mathbb{E}[\texttt{Q}]_{\text{nom}} = m(T_{\rm design}{+}1) = 62$ for $m=2$.
For Attack~A, a colored first-order autoregressive disturbance is injected on the velocity channels,
modeled as $\mathbf d_k = \omega\,\mathbf d_{k-1} + \psi\,\boldsymbol{\xi}_k$ with
$\boldsymbol{\xi}_k \sim \mathcal N(\mathbf 0, I_2)$.
In the experiments, $\omega=0.98$ and $\psi=8\times 10^{-3}$ are chosen so that the aggregate NIS
over the injection window is not rejected at significance $\alpha=0.05$, while the temporal
dependence remains detectable by the SDS whiteness test.

\color{black}

To assess residual effects after the disturbance ends, a longer
evaluation horizon of $T{+}1=301$ samples (i.e., $25$ minutes) is simulated, while injecting
the attack only over the design window.

For Attack~B, the FIR uses $r=2$ (i.e., three taps) and the SDS whiteness test uses $L=10$ low-lag moments.

For \(m=2\), \(T_{\rm design}=30\), and \(\alpha=0.05\), this yields
\(\varepsilon_{\rm cov}\approx 19.38\), which is adopted in the experiments.
\color{black}
\textcolor{black}{Moreover,} the disturbance energy bound $\mathrm{tr}(\Sigma)\le D$ is considered with $D=5.0$, and the whiteness constraint is chosen to be consistent with the same SDS setting ($L=10$, $\alpha=0.05$, i.e., the $\chi^2_{m^2L}$ threshold).

\textcolor{black}{All experiments were executed in Python using the CVXPY optimization framework, on a laptop with an
Intel\textsuperscript{\textregistered} Core\textsuperscript{\texttrademark} i7--9750H
processor and 32\,GB of RAM. Convex programs were solved using the SCS solver with fixed
tolerances ($\texttt{eps}=10^{-6}$) and iteration limits (up to $2\times10^{5}$ iterations)
across all experiments. The covariance‑fit recovery was implemented in PyTorch and optimized with Adam, initialized from the quadratic Toeplitz projection.}

\color{black}

To compute a feasible stealthy disturbance, \textcolor{black}{the} proposed pipeline is followed based on
convex relaxation, covariance-fit recovery, and final scaling.
First, the relaxed optimization problem is solved, obtaining a solution with objective
value $J_{\rm relax}\approx 1.85\times 10^{6}$.
This value represents the maximum achievable cost over the relaxed feasible set and
serves as a reference upper bound for the subsequent structured recovery.

Next, FIR realizability is enforced by solving the covariance-fit recovery problem\footnote{\textcolor{black}{The parameter $\zeta$ is set to zero in the recovery objective to avoid biasing the fit beyond the hard NIS/whiteness constraints, which are enforced afterward via the scaling step.}}
$\min_{\mathcal T\in\mathbb T_M}\|\mathcal T\mathcal T^\top-\Sigma^\star\|_F^2$ $\text{ s.t. } \|\mathcal T\|_F^2\le D$
(using a gradient-based optimizer, Adam, initialized from the quadratic Toeplitz projection).
The resulting unscaled solution attains $J_{\rm rec}^{\rm (unscaled)}\approx 4.97\times 10^{5}$,
but it violates the detection and energy bounds.

To ensure feasibility with respect to all detection and energy constraints, finally, the scaling step in Eq.~\eqref{eq:scaling} is applied, yielding a scaling factor
$\gamma\approx 0.181$.
After scaling, the recovered disturbance satisfies all \textcolor{black}{construction constraints}
and achieves a cost of $J_{\rm rec}\approx 1.62\times 10^{4}$.
Numerically, this corresponds to $J_{\rm rec}/J_{\rm relax}\approx 8.8\times 10^{-3}$.

In accordance with Theorem~\ref{thm:nontrivial-existence} and the bound
in~\eqref{eq:relax1}, the discrepancy between the relaxed solution
$\Sigma^\star$ and its Toeplitz-structured recovery $\widehat\Sigma$
(prior to scaling) satisfies Eq.~\eqref{eq:relax1}.
In the present experiment, the observed difference is
$|J_{\rm relax}-J_{\rm rec}^{\rm (unscaled)}|\approx 1.35\times 10^{6}$,
while the right-hand side evaluates to
$\|Q_H\|_2\,\|\Sigma^\star-\widehat\Sigma\|_*\approx 1.84\times 10^{7}$.
This large margin illustrates the conservative, worst-case nature of the
nuclear-norm bound when combined with structural recovery constraints.
As a result of the proposed approach, taps
$\{M_\ell\}_{\ell=0}^{2} \subset \mathbb{R}^{2\times 2}$ (units m/s per seed unit) are obtained with
\[
M_0\!=\!\!\begin{bmatrix}
5.12\!\times\!10^{-2} & 1.14\!\times\!10^{-5}\\
-1.58\!\times\!10^{-6} & -5.12\!\times\!10^{-2}
\end{bmatrix},
\]
\[
M_1\!=\!\!\begin{bmatrix}
1.27\!\times\!10^{-2} & 8.24\!\times\!10^{-6}\\
-5.97\!\times\!10^{-6} & -1.28\!\times\!10^{-2}
\end{bmatrix},
\]
\[
\begin{aligned}
M_2=\begin{bmatrix}
2.15\!\times\!10^{-3} & 7.82\!\times\!10^{-6}\\
-7.76\!\times\!10^{-6} & -2.17\!\times\!10^{-3}
\end{bmatrix}.
\end{aligned}
\]

\color{black}
The statistical evidence corresponding to the two attacks is reported in Figure~\ref{fig:pvalsAB}, which shows the per-test $p$-values (bias, NIS, Gaussianity, whiteness). 
Attack~A is characterized by near-zero NIS and whiteness $p$-values, confirming that temporal coloring is what betrays this spoof despite its mean being near zero and its marginal shape not being strongly non-Gaussian. 
In contrast, Attack~B keeps the NIS and whiteness $p$-values well away from the rejection region, indicating that the SDS accepts this realization.

\color{black}
In the attack window, some $p$‑values may exceed their nominal counterparts, reflecting the attacker's explicit shaping of innovation statistics under the SDS constraints; this can be observed, for instance, in the mean test. This effect should be interpreted as a consequence of constrained covariance/whiteness matching rather than as inherently ``cleaner" measurements.

Notably, although the FIR is optimized on a $30$--step window, its statistics remain non-rejected over the longer evaluation horizon. Because the disturbance is injected only over the $T_{\rm design}{+}1=31$-sample attack window and set to zero afterward, full-horizon statistics are diluted by the post-attack segment; consequently, full-horizon $p$-values can be larger than the corresponding attack-window values. For Attack~B (recovered FIR), the full-horizon $p$-values remain comfortably non-rejected (mean $0.586$, NIS $0.344$, Gaussianity $0.378$, whiteness $0.820$, min-$p$ $0.344$). Attack~A, instead, remains strongly detectable due to its pronounced temporal coloring (all $p$-values are below $0.0018$ except for Gaussianity, for which a  $p$-value equal to $0.365$ is obtained).

\color{black}

Figure~\ref{fig:dispAB} plots the displacement magnitude $\|\Delta {\bm p}[k]\|_2$ for the two attacks over time, where ${\bm p}[k]=[p_x[k],p_y[k]]^\top$ is the vessel position on the plane.
\color{black}
Attack~A exhibits a clear and growing drift consistent with strong colored excitation, reaching a peak displacement of about $1.13\times 10^{3}$\,m (at $t=1500$\,s), whereas Attack~B produces a much smaller but still noticeable deviation, peaking around $2.09\times 10^{2}$\,m (at $t=1375$\,s).

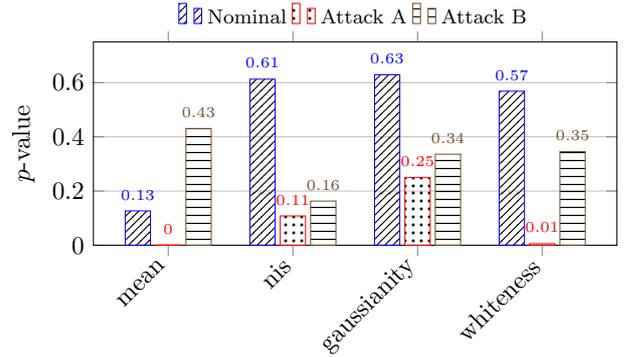
\begin{figure}[t]
  \centering
  \resizebox{0.95\columnwidth}{!}{
  \input{fig1.tex}
  }
  \vspace{-2mm}
  \caption{Per-test $p$-values for the nominal case, Attack~A and Attack~B, evaluated at $t=150 s$ (i.e., at the end of the injection horizon). }
  \label{fig:pvalsAB}
\end{figure}

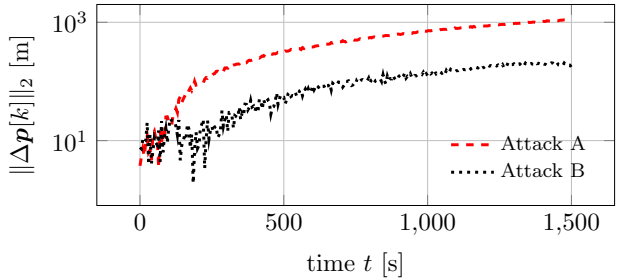
\begin{figure}[t]
  \centering
  \resizebox{0.95\columnwidth}{!}{
  \input{fig2.tex}
  }
  \vspace{-2mm}
  \caption{Displacement $\|\Delta {\bm p}[k]\|_2$ over time relative to nominal. Dashed: Attack A; dotted: Attack B.}
  \label{fig:dispAB}
\end{figure}

The maritime evaluation highlights two complementary aspects of the detection problem. 
When disturbances are colored, as in the AR(1) spoof of Attack~A, the trajectory quickly develops a visible drift. 
Although the aggregate NIS remains comparable to its nominal budget, the strong temporal correlation betrays the attack: both the mean and whiteness tests push the SDS to reject this pattern. 
In contrast, the recovered FIR disturbance of Attack~B demonstrates the opposite compromise. 
Here the deviation from nominal is smaller, yet statistically much harder to distinguish: the per-test $p$-values remain comfortably within the acceptance region. 
This attack therefore passes undetected while still biasing the trajectory in a subtle but noticeable way. 
Together, these results illustrate the practical damage–stealth trade-off: aggressive coloring yields larger effects but is detectable, while carefully shaped FIR disturbances retain stealth at the price of smaller impact. 
Most importantly, the case study confirms that considering multiple tests into the SDS decision substantially improves resilience over $\chi^2$ NIS-only checks, which would have missed the AR(1) spoof entirely.

\section{Conclusions}
\label{sec:conclusions}
This work presented an innovation-based diagnostic suite that unifies mean, NIS,
Gaussianity, and whiteness checks. The maritime evaluation illustrates its behavior
under realistic navigation conditions. Results show that standard $\chi^2$ monitoring
of the NIS alone can be bypassed by temporally colored disturbances, whereas the
proposed SDS reliably exposes such anomalies. Conversely, disturbances generated
via the relaxation–recovery pipeline demonstrate that stealth can be engineered,
yielding smaller yet statistically consistent deviations.

\color{black}
From a practical standpoint, the main advantage of the proposed SDS lies in its
lightweight, model-consistent integration within existing KF-based navigation
pipelines, while providing interpretable, per-property statistical evidence.
At the same time, the approach inherits intrinsic limitations of innovation-based
monitoring: detection is statistical (not deterministic), impact under stealth
constraints is necessarily bounded, and finite-horizon tests may lose sensitivity
to very slow or long-term deviations.
\color{black}

Together, these findings clarify both the robustness and the limits of innovation-based
monitoring in maritime estimation and point toward extensions to longer horizons,
richer dynamics, and integration with higher-level surveillance tools\textcolor{black}{. In doing so, they advance} the reliability and transparency of maritime anomaly detection in support
of safer and more resilient vessel operations.

\bibliography{references}

\end{document}

%% file: fig1.tex
\begin{tikzpicture}
\begin{axis}[
  width=\linewidth,
  height=0.5\linewidth,
  ybar=2pt,
  bar width=10pt,
  ylabel={$p$-value},
  symbolic x coords={mean,nis,gaussianity,whiteness,fused},
  xtick=data,
  xticklabel style={rotate=45, anchor=east},
  ymin=0, ymax=0.75,
  ymajorgrids,
  legend style={at={(0.5,1.02)},anchor=south,draw=none,fill=none,legend columns=3,font=\small},
  nodes near coords,
  every node near coord/.append style={
    font=\tiny,
    /pgf/number format/fixed,
    /pgf/number format/precision=2,
    yshift=1pt
  },
  enlarge x limits=0.20
]
\addplot+[pattern=north east lines] coordinates {(mean, 0.126710) (nis, 0.612805) (gaussianity, 0.629128) (whiteness, 0.568511)};
\addlegendentry{Nominal}
\addplot+[pattern=dots] coordinates {(mean, 0.0005310556) (nis, 0.1080924170) (gaussianity, 0.2503607363) (whiteness, 0.0055026161)};
\addlegendentry{Attack A}
\addplot+[pattern=horizontal lines] coordinates {(mean, 0.430830) (nis, 0.162932) (gaussianity, 0.335950) (whiteness, 0.345465)};
\addlegendentry{Attack B}

\end{axis}
\end{tikzpicture}

%% file: fig2.tex

\begin{tikzpicture}
\begin{axis}[
  width=\linewidth,
  height=0.5\linewidth,
  ymode=log,
ymin=0.8,
ymax=2e3,
ytickten={1,3,5},
  xlabel={time $t$ [s]},
  ylabel={$\|\Delta {\bm p}[k]\|_2$ [m]},
  grid=both,
  legend style={at={(0.97,0.07)},anchor=south east,draw=none,fill=none,font=\small},
  legend cell align=left
]

\addplot+[red,no marks, dashed, very thick]
  table[x=time_s, y=disp_attack_A_m, col sep=comma]
  {DATAattackA.csv};
\addlegendentry{Attack A}

\addplot+[black,no marks, dotted, very thick]
  table[x=time_s, y=disp_attack_C_m, col sep=comma]
  {DATAattackB.csv};
\addlegendentry{Attack B}

\end{axis}
\end{tikzpicture}